%% file: main.tex
\documentclass[sigplan,10pt,screen]{acmart}
\renewcommand\footnotetextcopyrightpermission[1]{}

\setcopyright{none}

\acmDOI{}
\acmISBN{}
\acmPrice{}

\PassOptionsToPackage{pdfpagelabels=false,pdftex,bookmarks=false,colorlinks=true,citecolor=blue,filecolor=black,linkcolor=blue,urlcolor=black}{hyperref}
\newcommand{\allnotes}[1]{}

\usepackage[utf8]{inputenc}
\usepackage{amsmath}
\usepackage{amsfonts}
\usepackage{amsthm}
\usepackage{enumitem}
\usepackage{geometry}
\usepackage{graphicx}
\usepackage{url}
\usepackage{xcolor}
\usepackage{tikz}
\usepackage{mdframed}
\usepackage[skip=5pt]{caption}
\usepackage[aboveskip=2pt]{subcaption}

\newcommand{\ata}{all-to-all }

\newtheorem{theorem}{Theorem}

\usepackage[textsize=tiny,textwidth=0.6in]{todonotes}
\renewcommand{\descriptionlabel}[1]{\hspace{\labelsep}\textit{#1}}

\newcommand{\captionfonts}{\bf \footnotesize}

\makeatletter  %
\long\def\@makecaption#1#2{%
	\vskip\abovecaptionskip
	\sbox\@tempboxa{{\captionfonts #1: #2}}%
	\ifdim \wd\@tempboxa >\hsize
	{\captionfonts #1: #2\par}
	\else
	\hbox to\hsize{\hfil\box\@tempboxa\hfil}%
	\fi
	\vskip\belowcaptionskip}
\makeatother   %

\newcommand{\squishlist}{
  \begin{list}{$\bullet$}{
    \setlength{\itemsep}{0pt}       \setlength{\parsep}{3pt}
    \setlength{\topsep}{3pt}        \setlength{\partopsep}{0pt}
    \setlength{\leftmargin}{1em}    \setlength{\labelwidth}{1em}
    \setlength{\labelsep}{0.5em} } }

\newcommand{\squishend}{
  \end{list} }

\newcommand{\squishenum}{
  \begin{enumerate}{}{
    \setlength{\itemsep}{0pt}       \setlength{\parsep}{3pt}
    \setlength{\topsep}{3pt}        \setlength{\partopsep}{0pt}
    \setlength{\leftmargin}{1em}    \setlength{\labelwidth}{1em}
    \setlength{\labelsep}{0.5em} } }

\newcommand{\squishenumend}{
\end{enumerate} }

\setlist[enumerate]{leftmargin=1.5em,topsep=3pt,itemsep=0pt}
\microtypecontext{spacing=nonfrench}

\begin{document}

\pagestyle{plain}

\title{Efficient \ata Collective Communication Schedules for Direct-connect Topologies}
\titlenote{Distribution Statement ``A'' (Approved for Public Release, Distribution Unlimited).}

\author{Prithwish Basu}
\affiliation{%
  \institution{\small{RTX BBN Technologies}}
  \streetaddress{10 Moulton St}
  \city{Cambridge}
  \state{MA}
  \country{}
  \postcode{02138}}
\email{prithwish.basu@rtx.com}

 \author{Liangyu Zhao}
\affiliation{%
  \institution{\small{University of Washington}}
  \streetaddress{185 E Stevens Way NE}
  \city{Seattle}
  \state{WA}
  \country{}
  \postcode{98195}}
\email{liangyu@cs.washington.edu}

 \author{Jason Fantl}
\affiliation{%
  \institution{\small{RTX BBN Technologies}}
  \streetaddress{10 Moulton St}
  \city{Cambridge}
  \state{MA}
  \country{}
  \postcode{02138}}
\email{jason.fantl@rtx.com}

 \author{Siddharth Pal}
\affiliation{%
  \institution{\small{RTX BBN Technologies}}
  \streetaddress{10 Moulton St}
  \city{Cambridge}
  \state{MA}
  \country{}
  \postcode{02138}}
\email{siddharth.pal@rtx.com}

\author{Arvind Krishnamurthy}
\affiliation{%
  \institution{\small{University of Washington}}
  \streetaddress{185 E Stevens Way NE}
  \city{Seattle}
  \state{WA}
  \country{}
  \postcode{98195}}
\email{arvind@cs.washington.edu}
 
 \author{Joud Khoury}
\affiliation{%
  \institution{\small{RTX BBN Technologies}}
  \streetaddress{10 Moulton St}
  \city{Cambridge}
  \state{MA}
  \country{}
  \postcode{02138}}
\email{joud.khoury@rtx.com}

\input{abstract}

\maketitle

\input{intro}
\input{model}

\input{algorithm}
\input{impl}
\input{eval}

\input{ack}

\bibliographystyle{ACM-Reference-Format}
\bibliography{refs}

\end{document}

%% file: abstract.tex
\begin{abstract}
The \ata collective communications primitive is widely used in machine learning (ML) and high performance computing (HPC) workloads, and optimizing its performance is of interest to both ML and HPC communities. All-to-all is a particularly challenging workload that can severely strain the underlying interconnect bandwidth at scale.
This paper takes a holistic approach to optimize the performance of \ata collective communications on supercomputer-scale direct-connect interconnects.
We address several algorithmic and practical challenges in developing efficient and bandwidth-optimal \ata schedules for any topology and lowering the schedules to various runtimes and interconnect technologies.
We also propose a novel topology that delivers near-optimal \ata performance. 
\end{abstract}

%% file: intro.tex
\section{Introduction}
Collective communications have received significant attention in both high performance computing (HPC) and machine learning (ML) disciplines. 
The \ata collective, in particular, is used in several HPC workloads such as with the 3D Fast Fourier Transform (FFT)~\cite{p3dfft} used in molecular dynamics~\cite{gromacs, berendsen1995gromacs} and direct numerical simulations~\cite{moin1998direct}.
It is also used in ML workloads, for example, to exchange large embeddings in the widely deployed Deep Learning Recommendation Model (DLRM)~\cite{naumov2020deep, mudigere2022software}, and in the mixture-of-experts (MoE) models~\cite{lepikhin2020gshard}. 
All-to-all collective communication is often a bottleneck at scale in these workloads~\cite{a2a-comm-bottleneck, a2a-comm-complexity, a2a-comm-bottleneck-ml}.

An emerging approach to meet these challenging demands has been to employ various forms of optical circuit switching to achieve higher bandwidths at reasonable capital expenditure and energy costs~\cite{sipml,wang2023topoopt,sip-switch,x-nest,hybrid-opt,jouppi2023tpu,lightwave-sigcomm}.
Hosts communicate using a limited number of optical circuits that may be reconfigured at timescales appropriate for the hardware (see \S\ref{subsec:direct-connect}), thus exposing network topology as a configurable component.
We refer to this setting as \emph{direct-connect} with circuits that are configured and fixed for an appropriate duration.
Direct-connect fabrics and topologies such as mesh, Tori, DragonFly~\cite{kim2008technology}, and SlimFly~\cite{besta2014slim} have been well studied in the HPC community and deployed across several supercomputers, such as with Google's TPUv4~\cite{jouppi2023tpu, lightwave-sigcomm}. 

Computing bandwidth-optimal \ata schedules on a direct-connect topology with $N$ nodes can be formulated using the Max Concurrent Multi-Commodity Flow problem, hereafter MCF, and solved in polynomial time using linear programming (LP)~\cite{shahrokhi1990maximum}. 
MCF, however, suffers from high time complexity even at modest scales since the number of flow variables in a bounded degree network scales as $\mathcal{O}(N^3)$. %
At $N=1000$, for example, even a state-of-the-art LP solver~\cite{mosek} is unable to generate a schedule on a fast machine within an entire day. For smaller $N (<100)$, which is typical of ML applications, it takes tens of minutes to generate a schedule. 
This makes it hard for the algorithm to react quickly to changes in the topology, for example, due to topology reconfiguration or failures.
{\em We enhance the scalability of the exact \ata MCF by decomposing it into a simpler master LP and a set of $N$ children LPs that are parallelized for fast computation}.
We demonstrate a $\mathcal{O}(poly(N))$ speed up in time complexity under decomposition and parallelization, reducing actual runtime on $N=1000$ by orders of magnitude to 40 minutes instead. For $N$ in the hundreds, it takes seconds to generate a schedule. 
Prior works~\cite{shahrokhi1990maximum, leighton1989approximate, fleischer2000approximating, karakostas2008faster} try to improve computational complexity by trading off optimality using approximation schemes.
These works still end up significantly underperforming our decomposed MCF in practice, both in terms of performance and complexity.

Another challenge lies in lowering the MCF solution to both ML accelerators and HPC runtimes and fabrics.
These fabrics employ different topology, routing, and flow control mechanisms as they have historically been designed with different objectives~\cite{dally2004principles}.
We devise a general model of the underlying network, distinguishing between fabrics that support additional forwarding bandwidth (i.e., forwarding bandwidth at the Network Interface Card (NIC) is higher than the injection bandwidth at the host/accelerator) and those that do not.
Additional forwarding bandwidth increases \ata performance in direct-connect settings as it compensates for the {\em bandwidth tax}~\cite{mellette2020expanding} (since a node acts as a router and uses a significant fraction of its total link bandwidth to forward other node traffic). 
We {\em develop an algorithmic toolchain for producing and lowering near bandwidth-optimal \ata collective communication schedules to arbitrary supercomputer-scale topologies and different interconnect technologies}.
On host or accelerator runtimes where data movement is ``scheduled'', we devise a novel time-stepped version of the MCF problem. %
On fabrics with hardware ``routing'' and additional forwarding bandwidth, we develop scalable algorithms for computing static routes either by directly extracting the paths from the MCF solution or by employing path-based MCF formulations where flow variables are defined on paths instead of on links. 
We develop compilers and tools for lowering the schedules and the routes to the underlying runtime and interconnect, and we demonstrate near-optimal \ata performance on a range of topologies at different scales.

Finally, {\em we establish an analytical lower bound for \ata performance on any topology, use it to compare different topologies and show the superiority of generalized Kautz graphs in terms of both performance and coverage}.
It is known that topologies with higher bisection bandwidth result in higher \ata throughput~\cite{besta2014slim, jouppi2023tpu}.
Several works in the HPC community have investigated the \ata behavior of different topologies. 
Earlier works proposed specialized patterns for higher dimensional mesh, tori~\cite{suh1998all} and hypercubes~\cite{johnsson1989optimum}, while later works proposed more complex topologies that have beneficial graph properties, e.g., high expansion coefficient~\cite{valadarsky2016xpander}, large spectral gap~\cite{young2022spectralfly}, and low diameter~\cite{besta2014slim}. 
Many of the proposed topologies, however, do not have sufficient {\em coverage} in realizable graph sizes $(N)$ and degree $(k)$. 
We propose the class of generalized Kautz (GenKautz) graphs~\cite{genkautz}, which are known for their expansion properties and can be constructed for any $N$ and $k$.

%% file: model.tex
\section{Background and Terminology}
\subsection{Direct-Connect Fabrics for ML and HPC}\label{subsec:direct-connect}
Our work identifies topologies and schedules helpful for a broad range of direct-connect interconnects common to both HPC and ML accelerator fabrics.
These include, for example, \emph{switchless physical circuits}~\cite{rockportnic}, \emph{patch-panel optical circuits}~\cite{patchpanel}, and \emph{optical circuit switches} (OCS)~\cite{jouppi2023tpu}. 
These options differ in cost, scalability, and reconfigurability~\cite{topoopt}. 
For example, commercially available OCSs can perform reconfigurations in $\approx$10ms, are more expensive than patch panels, but scale to fewer ports (e.g., Polatis 3D-MEMS switch has 384 ports at \$520 per port~\cite{polatis}). 
With these reconfigurable fabrics, topology becomes a degree of freedom, and ongoing work is demonstrating how to exploit this degree of freedom for increased performance~\cite{jouppi2023tpu, sip-switch,sipml,wang2023topoopt, zhao2022optimal}.
Despite supporting faster reconfigurations,  OCSes still suffer from relatively high reconfiguration latency, precluding rewiring of the circuits during a typically-sized collective operation. 
Accordingly, collectives need to operate over a set of pre-configured circuits that remain unchanged for the duration of the collective operation.
We refer to this setting as \emph{direct-connect}, circuits (and topology) that are configured and remain static for the duration of the collective algorithm.

Our work additionally targets different interconnect technologies, broadly ML accelerator and HPC interconnects. 
These employ different topologies, routing, and flow control, as they have historically been designed with different objectives~\cite{dally2004principles, naumov2020deep}.
Table~\ref{tbl:comparison_fabric} highlights high-level differences between the two fabrics.
HPC interconnects have generally focused on reducing latency using low-diameter topologies with high bisection bandwidth and hardware routing with cut-through flow control. 
With hardware routing, where each node or NIC serves as a router, the total forwarding bandwidth may exceed the host injection bandwidth to accommodate for the forwarding bandwidth tax.
ML accelerator interconnects, on the other hand, optimize for high link bandwidth as they are mostly focused on collectives, tend not to employ hardware routing, and use synchronized accelerator schedules with store-and-forward flow control. 
\begin{table}[b]
\small
\caption{Comparison of HPC and ML accelerator fabrics.}
\begin{center}
\scalebox{0.9}{
\begin{tabular}{|c|c|c|}
\hline 
& {\bf HPC} & {\bf ML Accelerator}\\ \hline
{\bf Schedules} & Path-based & Link-based \\ \hline
{\bf Topology focus} & Bisection bandwidth & Node bandwidth \\ \hline
{\bf Flow Control} & Cut-through & Store-and-forward \\ \hline
{\bf Injection BW} & $B$ & $B$ \\ \hline
{\bf Forwarding BW} & $\ge B$ & $B$ \\ \hline
\end{tabular}
}
\end{center}
\label{tbl:comparison_fabric}
\end{table}%
\subsection{All-to-all Schedules, and Throughput}\label{sec:schedules_throughput}
The network topology is modeled as a directed graph, represented as the tuple $G=(V,E)$, where $V$ denotes the set of nodes ($|V|=N$) and $E$ denotes the set of directed edges. 
The direct-connect fabric imposes a constraint that all nodes have degree $d$, which is the number of links/ports on each host or accelerator and is ideally low and independent of $N$. 
The link bandwidth is $b$, and the node bandwidth is $B=db$.

Each node $i$ has a data buffer $B_i$ comprised of $N$ contiguous and equally sized shards $B_{i,j}$ each of size $m$ bytes, $0 \leq i,j < N$, $|B_i|=Nm$, $|B_{i,j}|=m$.
The \ata collective \textit{transposes} the buffers, i.e., each node $i$ sends shard $B_{i,j}$ to node $j$. %

Communication schedules can operate at a finer granularity than a shard. 
We define chunk $C_{i,j}$ to be a subset of shard $B_{i,j}$, both specified as \emph{index sets} of elements in a shard with $B_{i,j}$ representing the entire shard. 
For example, the shard can be an interval $[0,1]$, and $C_{i,j}$ be some subinterval. 
Chunks do not need to be the same size.
Since each chunk $C_{i,j}$ has a known source node $i$ and destination node $j$, we omit the indexes and simply use $C$ to denote the chunk.
An \ata {\em communication (comm) schedule} $A$ for $G$ with $t_{\max}$ {\em comm steps} specifies which chunk is communicated over which link or route in any given {\em step}. 
Specifically, $A$ is a set of tuples $(C,(u,w),t)$ with $u,w\!\in\! V$ and $t\!\in\!\{1,\dots,t_{\max}\}$.
$(C,(u,w),t)$ denotes that node $u$ sends chunk $C$ to node $w$ at comm step $t$.
Chunking is performed during schedule compilation (\S\ref{sec:compiler}).

\noindent{\bf Link-based Schedules:} In fabrics without hardware routing, chunks only flow on directly connected edges $(u,w)\!\in\! E_G$. 

\noindent{\bf Path-based Schedules:} In fabrics with hardware routing, $(u,w)$ may not correspond to an edge in $G$, i.e., chunks can flow on end-to-end paths between source and destination as determined by the routing function. 

The {\em throughput} of an \ata schedule for a shard size $m$ is $\frac{(N-1)m}{T}$, where $T$ is the time to complete the \ata schedule (the time for each node to send $N-1$ shards each of size $m$ bytes). 

Finally, %
 {\em algorithm runtime} is the time taken by the {\em algorithm} to compute and lower the schedule for a given network.

\input{related}

%% file: related.tex
\vspace{-2ex}
\subsection{Related work}\label{sec:related}
In the theory community, optimization of the \ata collective has been formulated as a maximum concurrent multi-commodity flow problem (MCF) and solved in polynomial time using LP~\cite{shahrokhi1990maximum}. Although the MCF has polynomial time complexity, it can be hard to solve in practice for large problem sizes. As a result, several works have proposed fully polynomial time approximation schemes (FPTAS)~\cite{shahrokhi1990maximum, fleischer2000approximating, karakostas2008faster}. The best known FPTAS schemes~\cite{karakostas2008faster} have time complexity 
$\mathcal{O}\left( \frac{N^3}{\epsilon ^2} \log^{\mathcal{O}(1)}N  \right)$ and get to a factor of $(1-\epsilon)$ of the optimal throughput. In this paper, we improve the tractability of LP-based solutions while not sacrificing optimality. We decompose the original MCF problem into a master LP and $N$ simpler parallelizable child LPs. Since the former (which dominates the time complexity -- see Fig.~\ref{fig:comptime_genkautz}) has $\mathcal{O}(N^2)$ variables, one can leverage recent LP solving techniques with time complexity $\mathcal{O}(N^{2.37})$~\cite{Cohen19} to solve the MCF in $\mathcal{O}(N^{4.74})$ time. In practice, our master LP has lower time complexity owing to its special structure, and MCF is significantly better in running time than the FPTAS schemes (for small values of $\epsilon$) without sacrificing optimality even for moderate $N$ (Fig.~\ref{fig:comptime_genkautz}). Moreover, the \emph{sequential} FPTAS schemes are unable to exploit the parallelism the way we do. %

Early HPC works investigated efficient \ata collective communication on well-known topologies, e.g., hypercubes, meshes, and tori. Johnsson and Ho~\cite{johnsson1989optimum} proposed optimal \ata collectives for single-port and $n$-port models of hypercubes. Scott~\cite{scott1991efficient} proposed optimal \ata collectives on meshes. Suh et al.~\cite{suh1998all} and Yang et al.~\cite{yang1999efficient} proposed \ata collectives for mesh and tori that could leverage virtual cut-through and wormhole-switched networks.

More recent works have studied \ata communication on topologies that have beneficial graph properties for supporting datacenter communications. The bisection bandwidth of a network ($\chi$) is known to be related to \ata throughput in the sense that the latter is bounded from above by $\frac{4\chi}{N^2}$.
Prior works have therefore used $\chi$ as a proxy for \ata throughput~\cite{valadarsky2016xpander,singla2012jellyfish,besta2014slim}, and
as a result, expander graphs got significant interest due to their low modularity and hence high $\chi$.
Xpander~\cite{valadarsky2016xpander} routes \ata traffic along K-shortest paths on expander graphs with multi-path TCP congestion control~\cite{mptcp} to yield good throughput in switch-based datacenter settings.
The \ata problem has been formulated as an MCF in such contexts~\cite{jyothi2016measuring,poutievski2022jupiter}, and it has been shown that
multiple expanders have nearly identical performance for \ata traffic.
However, ours is the first study that applies multiple forms of MCF constructs (link- and path-based) to optimize \ata collective communications on a diverse set of HPC and ML fabrics and topologies at scale.

Recently, Cai et al.~\cite{SCCL} proposed an SMT-logic-based approach (SCCL) for synthesizing optimal collectives in a topology-agnostic manner for GPU fabrics. However, their approach is computationally expensive due to the NP-hard nature of the SMT formulation. Followup work TACCL~\cite{TACCL} relies on integer programming and suffers from similar computational bottlenecks, as we show in \S\ref{sec:eval}.
Recently proposed TE-CCL~\cite{arzani2023rethinking} improves upon TACCL's performance by combining multi-commodity flow with Mixed Integer Linear Programming (MILP) and A* search. 
Their models focus on link-driven latency, which can be important at small sub-Megabyte buffer sizes. 
Our formulations, on the other hand, maximize network utilization for \ata under large buffer sizes, and we observe that 
MCF solutions in general attempt to take short paths through the network anyway. 
Our approach is significantly more scalable, generating efficient schedules for 1K+ nodes in much less time than what TE-CCL reports it takes to solve \ata on 128 node networks.
Finally, the work in~\cite{zhao2022efficient} optimizes the all-reduce collective.

%% file: algorithm.tex
\vspace{-2ex}
\section{Multi-commodity Flow-based Algorithms}
\begin{figure}
\centering
\includegraphics[width=0.9\columnwidth]{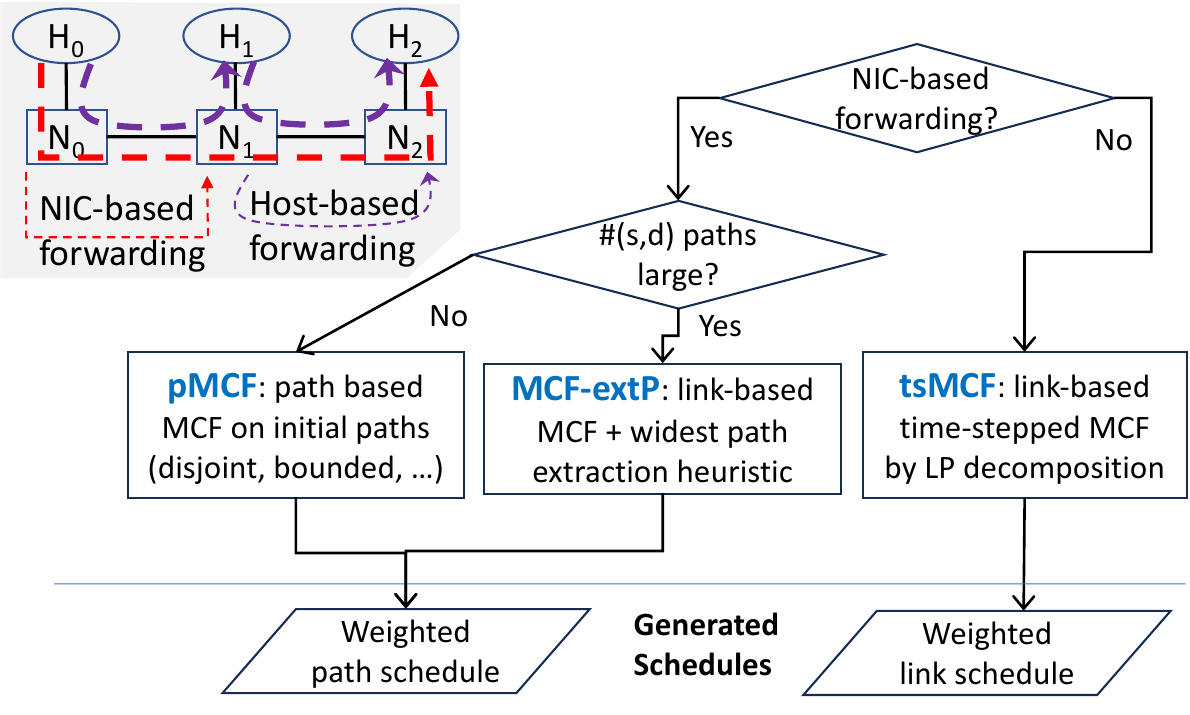}
\caption{\label{fig:flowchart} Generating link- and path-based schedules. Top left example shows difference between NIC-based and host-based forwarding of flow from host $\text{H}_0$ to $\text{H}_2$.}
\end{figure}
Fig.~\ref{fig:flowchart} shows a summary flowchart of the algorithms we employ for generating \ata\ schedules for direct-connect fabrics.
For ML-style fabrics with host/GPU-based forwarding,
we generate weighted link-based schedules by solving the time-stepped version of the MCF. %
An MCF solution defines what chunks of data corresponding to a certain $(s,d)$ pair (or commodity) should be transmitted by an intermediate node $u$ over each of its outgoing links $(u,v)$ at time step $t$. 
A naive solution %
involves solving a linear program (LP) on variables defined for each commodity, link, and time step--in the worst case, the total number of variables grows as $N(N-1)\times O(N) \times O(N) = O(N^4)$ where $N$ is the network size (bounded degree networks have $O(N)$ links and the number of time steps, $l_{max} \geq$ the diameter, which can be $O(N)$). 
We propose to {\em decompose} this LP into a {\em master} source-only LP that first computes aggregate optimal flow rates leaving each source $s$ and then uses this solution to compute optimal flow rates for each $(s,d)$ pair. 
This enables scaling to networks with thousands of nodes.

For HPC-style fabrics with NIC-based forwarding,
we generate path-based schedules that constitute a set of paths ${\mathcal P}_{s,d}$ for each $(s,d)$ pair and weights $w_{p_i}$ associated with each path $p_i\in{\mathcal P}_{s,d}$ controlling the fraction of traffic that should be sent along $p_i$. %
Optimal path-based schedules can be computed by solving the path-based version of the MCF, which is a natural \textit{dual} of the link-based version mentioned earlier. 
However, this involves defining optimization variables for every possible $(s,d)$ path, which is prohibitive for many topologies, even if we restrict the path set to include only shortest paths. %
We use good heuristics like sampling good path sets of small cardinality (e.g., edge-disjoint paths) to mitigate this problem. 
We also propose another radically different approach that instead solves the link-based MCF, and then applies an iterative ``widest path'' extraction algorithm to greedily extract high-flow $(s,d)$ paths from the optimal per-link flows. 
Although potentially suboptimal, this approach is tractable and has good performance on the topologies we study. 
\vspace{-2ex}
\subsection{Problem formulation}
\label{subsec:LP-basic}
\subsubsection{Link variable based MCF formulation}
\label{subsec:LP-basiclink}
Given a network $G=(V,E,cap:E\to\mathbb{R}^+)$, where $cap$ denotes link capacities, the problem of maximizing \ata throughput can be modeled as a \textit{maximum concurrent multi-commodity flow} (MCF) problem with $N(N-1)$ commodities %
of equal demand. This problem can be formulated using Linear Programming~\cite{shahrokhi1990maximum, leighton1989approximate, fleischer2000approximating, karakostas2008faster}. We define variables $f_{(s,d),(u,v)}$ to denote the amount of flow of commodity $s\to d$ that should traverse link $(u,v)$ and {\em concurrent} demand variable $F$ (i.e., the common rate at which all commodities will flow concurrently), and solve the LP below.
\begin{small}
\begin{center}
\noindent\textbf{\underline{Link-based max-concurrent MCF formulation:}}
\end{center}
\vspace{-2ex}
\begin{align}
&\mbox{maximize }  F \label{eq:mcmcf-F} \\
&\mbox{subject to:} \sum_{s,d} f_{(s,d),(u,v)} \leq cap_{(u,v)}, \forall u,v \label{ineq:mcmcf-capacity} \\
&\sum_{v} f_{(s,d),(u,v)} \displaystyle\leq \sum_{w} f_{(s,d),(w,u)}, \forall s,d,u:s\neq u,d\neq u \label{ineq:mcmcf-flow-conservation} \\
&\sum_{w} f_{(s,d),(w,d)} \geq F, \forall s,d \label{ineq:mcmcf-demand-conservation} \\
&f_{(s,d),(u,v)} \geq 0, \forall s,d,u,v \label{ineq:mcmcf-non-negativity}
\end{align}
\end{small}
The flow conservation constraint is modeled by inequality~\eqref{ineq:mcmcf-flow-conservation}. 
This improves the speed of the LP solver; at the optimal solution, the inequality is enforced with no slack. %
Also, enforcing the demand constraint~\eqref{ineq:mcmcf-demand-conservation} only at the sink node $d$ is sufficient since the combined flow conservation and demand constraints at the sink enforce the same at the source. If however, a flow $f_{(s,d)}$ with optimal $F$ returned by the solver has extra flow near $s$ (due to inequality \eqref{ineq:mcmcf-flow-conservation}), a post-processing step from $d$ to $s$ is executed to ensure exact flow conservation. An optimal flow generally follows links along multiple paths over the network.
This LP is solvable in polynomial time, albeit in high-order polynomial time. To improve solver efficiency, we use a compact formulation of the LP in which all the flow conservation and demand constraints are expressed by a single matrix-vector constraint that relates the product of the node-to-link incidence matrix and link-flow vector to the per-commodity demand matrix scaled by $F$. This eliminates the ``pre-solve'' canonicalization step.

A key disadvantage of the per-commodity based LP approach is that the number of link-flow variables for a $k$-regular graph is $k N^2 (N-1)$, which gets intractable even at modest scales (hundreds of nodes).
\subsubsection{Decomposing the MCF LP for scalability}
\label{subsec:LP-decomposition}
Since the MCF LP approaches discussed above are computationally challenging, we decompose the problem of computing the optimal flow for $N(N-1)$ commodities by considering $N$ groups of source-rooted flows (each delivered to $N-1$ destinations). Specifically, we follow the steps below.

\noindent\emph{(1) Compute source-based grouped commodity flows:} Solve a master LP, defined in \eqref{eq:Dmcmcf-source-F}-\eqref{ineq:Dmcmcf-source-non-negativity}, for computing the optimal concurrent rates for $N$ source-rooted grouped multicommodity flows.
The source-based flow conservation~\eqref{ineq:Dmcmcf-source-flow-conservation} reflects the fact that the total amount of flow entering $u$ has to be greater than the sum of the amount of flow leaving $u$ and the amount sunk at $u$ (which must equal the concurrent flow value $F$).
Since we worry about only $N$ groups of commodities (instead of $N(N-1)$ point-to-point commodities), we only need $k N^2$ variables, which is tractable (thousands of nodes). 

\noindent\emph{(2) Compute optimal per-commodity flows:} Once the per-source optimal flow has been computed, solve $N$ additional simpler Child LPs, one per source as defined in \eqref{eq:Dmcmcf-s-f}-\eqref{ineq:Dmcmcf-s-non-negativity}, to determine the flow values per link for each $(s,d)$ commodity. 
Each such LP (say for source $s$) will set the link capacities to the flow values computed by the master LP, and will solve a standard maximum concurrent multicommodity flow (on a thusly capacity-adjusted graph) for $N-1$ commodities $\{s\to v | v \in V\setminus\{s\}\}$. These LPs can be run in parallel on multi-core processors, have $k N(N-1)$ flow variables to optimize, and are generally simpler in complexity than the original LP. 
Solving $N + 1$ LPs with $O(N^2$) variables each is much more tractable than solving a single LP with $O(N^3)$ variables since the computation complexity of LP is generally much higher than linear in the number of variables.
\newpage
\begin{small}
\begin{center}
\noindent\textbf{\underline{Decomposed link-based MCF (for scalability):}}
\end{center}
\emph{Master LP to compute source-based grouped commodity flows:}
\begin{align}
&\mbox{maximize }  F \label{eq:Dmcmcf-source-F} \\
&\mbox{subject to:} \sum_{s} f'_{s,(u,v)} &\leq& \quad cap_{(u,v)}, \forall u,v \label{ineq:Dmcmcf-source-capacity} \\
&F + \sum_{v} f'_{s,(u,v)} &\leq& \quad \sum_{w} f'_{s,(w,u)}, \forall s,u:s\neq u \label{ineq:Dmcmcf-source-flow-conservation} \\
&f'_{s,(u,v)} &\geq& \quad 0, \forall s,u,v \label{ineq:Dmcmcf-source-non-negativity}
\end{align}
\emph{Child LPs to extract link flows from source-based flows:}
\begin{align}
&\mbox{minimize }  \sum _{d,u,v} f_{(s,d),(u,v)} \label{eq:Dmcmcf-s-f} \\
&\mbox{subject to:} \sum_{d} f_{(s,d),(u,v)} \quad \leq \quad f'_{s,(u,v)}, \forall u,v \label{ineq:Dmcmcf-s-capacity} \\
&\sum_{v} f_{(s,d),(u,v)} \quad \leq \quad \sum_{w} f_{(s,d),(w,u)}, \forall d,u:s\neq u,d\neq u \label{ineq:Dmcmcf-s-flow-conservation} \\
&\sum_{w} f_{(s,d),(w,d)} \quad \geq \quad F, \forall d \label{ineq:Dmcmcf-s-demand-conservation} \\
&f_{(s,d),(u,v)} \quad \geq \quad 0, \forall s,d,u,v \label{ineq:Dmcmcf-s-non-negativity}
\end{align}
\end{small}
The decomposed LP approach yields the optimal MCF value $F$ as the standard approach (in \S\ref{subsec:LP-basiclink}) although the actual flow values $f$ returned may be different.

\subsubsection{Time-stepped MCF (tsMCF) formulation}
\label{subsec:LP-timestepped}
The MCF formulation presented in \S\ref{subsec:LP-basiclink} yields the optimal rates at which each commodity should be transmitted by considering the flow of data to mimic that of infinitesimally divisible fluids. This is inadequate for ML network fabrics where accelerators send finite data chunks in a finite number of fixed-length time steps. This necessitates the generation of time-stepped schedules. To this end, we extend the notion of MCF to the temporal domain by computing flows on a time-expanded stacked graph~\cite{temporalG} representation of $G$. It has $l_{\max}+1$ time-indexed instances $u_t$ of each node $u\in G$ and directed edges $u_t \to v_{t+1}$ with $capacity=1$ whenever $(u,v)\in G$ as well as ``self" edges $u_t \to u_{t+1}$ with $capacity=\infty$ denoting potential buffering at $u$ over time. $l_{\max}$ is set to a value $\geq diameter(G)$. We compute flows on this time-expanded graph essentially following the procedure described in \S\ref{subsec:LP-basiclink}. The main difference is in the size of the input graph. In this case, it has $(l_{\max}+1) |V|$ nodes and $l_{\max} (|V| + |E|)$ links, and hence the LPs take somewhat longer to solve. However, the time-expanded graphs are directed acyclic graphs and are hence less complex. This tends to help mitigate the running time issues of the LPs. Another key difference is that all nodes do not source/sink traffic; instead, the nodes $u_0$ source traffic for nodes $v_{l_{\max}+1}$ and non-zero flow along an infinite capacity ``self'' edge essentially simulates waiting for a time slot.
The equivalent time-stepped LP formulation is given below.
\newpage
\begin{small}
\begin{center}
\noindent\textbf{\underline{tsMCF: for ML fabrics with host/GPU forwarding}}
\end{center}
\vspace{-2ex}
\begin{align}
&\mbox{minimize } \sum _t U_t \label{eq:mcmcf-Ut} \\
&\mbox{subject to:} \sum_{s,d} f_{(s,d),(u,v),t} \leq U_t, \forall u,v,t \label{ineq:mcmcf-t-capacity} \\
&\sum_{t' \leq t, v} f_{(s,d),(u,v),t'} \leq \sum_{t'' < t,w} f_{(s,d),(w,u),t''},\notag\\
&\qquad\qquad\qquad\forall s,d,u,t:s\neq u,d\neq u, t > 1 \label{ineq:mcmcf-t-flow-conservation} \\
&\sum _{t',v} f _{(s,d),(u,v),t'} = \sum _{t'',w} f_{(s,d),(w,u),t''}, \forall s,d \label{eq:mcmcf-t-sd} \\
&\sum_{t',v} f_{(s,d),(s,v),t'} = \sum_{t'',w} f_{(s,d),(w,d),t''}=1, \forall s,d \label{ineq:mcmcf-t-demand-conservation} \\
&0 \leq f_{(s,d),(u,v),t} \leq 1, \forall s,d,u,v \label{ineq:mcmcf-t-non-negativity}
\end{align}
\end{small}
The objective \eqref{eq:mcmcf-Ut} minimizes the utilization of bandwidth at every time step, while the constraint~\eqref{ineq:mcmcf-t-capacity} ensures that the total utilization at every edge is less than the bandwidth. The constraint~\eqref{ineq:mcmcf-t-flow-conservation} enforces the amount of data received by node $u$ must be greater than or equal to the amount of data sent by node $u$ from comm step $1$ to every other comm step (the difference is reserved for future send). Constraint~\eqref{eq:mcmcf-t-sd} enforces the total amount received by node $u$ is equal to the total amount sent by node $u$. Constraint~\eqref{ineq:mcmcf-t-demand-conservation} enforces that the total amount sent by the source and received by the destination is equal to 1. One can add a multiplier to $f_{(s,d),(u,v),t}$ in the first constraint if link bandwidth is not uniform among all links.
This time-stepped LP can be decomposed into a source-based LP + child LPs as described in \S\ref{subsec:LP-decomposition}.
\subsubsection{Path-variable based MCF (pMCF) formulation}
\label{subsec:LP-basicpath}
In networks supporting multi-hop routing, we need to compute the optimal rates of flows along multiple paths from each source $s$ to each target node $d$. 
We first compute a set of paths ${\mathcal P}$ %
for each commodity/$(s,d)$ pair. 
Next, for each path $p\in {\mathcal P}$ and $(s,d)$ pair, define flow variable $f_{(s,d), p} \in \mathbb{R}^+ \cup \{0\}$. 
As in the link-based formulation, we ensure that the link capacity constraints are obeyed at each link. 
The flow conservation constraints are automatically obeyed at each node since data will be flowing along simple paths (in-degree and out-degree are 1).
The LP formulation is shown in (\ref{eq:mcmcf-path-F})-(\ref{ineq:mcmcf-path-non-negativity}).

If the set ${\mathcal P}$ is allowed to consist of all paths of unbounded length, then path-based MCF is a natural dual of link-based MCF and hence provides the same optimal MCF value. 
However, solving such a dual problem is impractical since $|{\mathcal P}|$ typically grows exponentially with $N$. 
We make the path-based MCF in practical scenarios tractable by curtailing the number of paths in ${\mathcal P}$ to $O(poly(N))$. 
Restricting the path lengths to below $l_{\max}$ drastically reduces $|{\mathcal P}|$. 
This approach works for many networks of interest (per our empirical observation), e.g., expander graphs like Generalized Kautz graphs, where $|{\mathcal P}|$ is polynomial in $N,l_{\max}$. However, in several other graphs that possess a high degree of symmetry, e.g., the torus, the number of paths of length $l_{\max}$ grows exponentially in $l_{\max}$ since the diameter is $\sqrt{N}$; therefore $|{\mathcal P}| \geq \frac{1}{\sqrt{N}+1} {2\sqrt{N} \choose \sqrt{N}}$, which grows super-exponentially in $N$; this makes the approach intractable for large tori. One tractable heuristic that we have empirically observed to achieve the optimal MCF solution is to choose ${\mathcal P}$ to be a maximal set of link-disjoint $(s,d)$ paths, which can be found efficiently and whose cardinality is upper bounded by $k N(N-1)$ for $k$-regular graphs.
\begin{small}
\begin{center}
\noindent\textbf{\underline{pMCF: for fabrics with NIC forwarding}}
\end{center}
\vspace{-2ex}
\begin{align}
&\mbox{maximize }  F \label{eq:mcmcf-path-F} \\
&\mbox{subject to:} \sum _{s,d} \sum_{p \ni e} f_{(s,d),p} \leq cap_{e}, \forall e \in E \label{ineq:mcmcf-path-capacity} \\
&\sum_{p \ni {\mathcal P}_{(s,d)}} f_{(s,d),p} \geq F, \forall s,d \label{ineq:mcmcf-path-demand-conservation} \\
&f_{(s,d),p} \geq 0, \forall s,d,u,v \label{ineq:mcmcf-path-non-negativity}
\end{align}
\end{small}
\vspace{-3ex}
\subsection{Applying MCF to the different fabrics}
\subsubsection{Source-routed fabrics}
\label{subsec:source-routed-fabrics}
The paths along which (certain fractions of) each commodity would flow must be provided at the respective sources of the commodity. We propose two different approaches below based on whether a topology has low or high path diversity (see Fig.~\ref{fig:flowchart}).

\noindent\textbf{\emph{pMCF: Directly applying path-variable based MCF (low path diversity).}}
The path-variable based MCF formulation described in \eqref{eq:mcmcf-path-F}-\eqref{ineq:mcmcf-path-non-negativity} can be applied to source-routed fabrics for graphs in which the path diversity does not grow exponentially in $N$, as is the case with expander graphs. 
However, this approach is not tractable for graphs like the multi-dimensional torus where the number of paths (with length $\leq l_{\max}$) grows super-exponentially in $N$. 

\noindent\textbf{\emph{MCF-extP: Applying link-based MCF and extracting paths (high path diversity).}}
We first solve the decomposed link-based MCF described in %
\S\ref{subsec:LP-decomposition}.
However, for source-routed fabrics, we need to obtain the paths on which the different commodities should flow. 
Therefore, as a final step, we greedily extract paths from the link-based MCF solution.

\noindent\textbf{\emph{Widest path extraction.}} Given the MCF flow values on each link corresponding to each $(s,d)$ commodity, we construct a weighted subgraph DAG $G_{s,d}$ of the original graph $G$ induced by the edges with non-zero $(s,d)$ flow (weights = MCF flows). We then iteratively extract $(s,d)$ paths from $G_{s,d}$ by greedily solving the \textit{widest path} problem, by making minor modifications to Dijkstra's shortest path algorithm:
\vspace{-1ex}
\begin{enumerate}
\item Find $(s,d)$ path $p$ in $G_{s,d}$ with maximum flow rate ($r$).
\item Subtract  $r$ from the capacities of all the links in $p$.
\item Repeat the two previous steps until $s$ no longer has a non-zero capacity path to $d$. 
\item Upon termination, the algorithm finds a set of $(s,d)$ paths with decreasing flow rates, which are ready for lowering.
\end{enumerate}
\subsubsection{Non source-routed fabrics (tsMCF)}
\label{subsec:non_source-routed-fabrics}
When forwarding and flow control are performed by the host/GPU (no hardware routing), %
we use the time-stepped link-based MCF formulation described in \S\ref{subsec:LP-timestepped} 
to obtain chunk schedules that exactly specify what data needs to be sent (or received) by each GPU at every time step. 

\noindent\textbf{\emph{Handling host-to-NIC bottlenecks}}
In scenarios where the host-to-NIC bandwidth $B_{host}$ is less than the net egress/ingress link bandwidth at the NIC, i.e., $B_{host} < d \cdot b$, the host-to-NIC bandwidth becomes a bottleneck.
Figure~\ref{fig:with_fw_bw_bottlenecked} shows how to augment the original NIC topology to model this host-to-NIC bottleneck. 
We augment the graph in a manner that forces data to flow through the host even when the node is not the destination. 
The MCF computed between the host nodes on the augmented graph yields the optimal throughput. 
We apply \S\ref{subsec:LP-timestepped}'s tsMCF formulation to generate schedules.
\vspace{-2ex}
\begin{figure}[htbp]
\centering
\includegraphics[width=0.8\columnwidth]{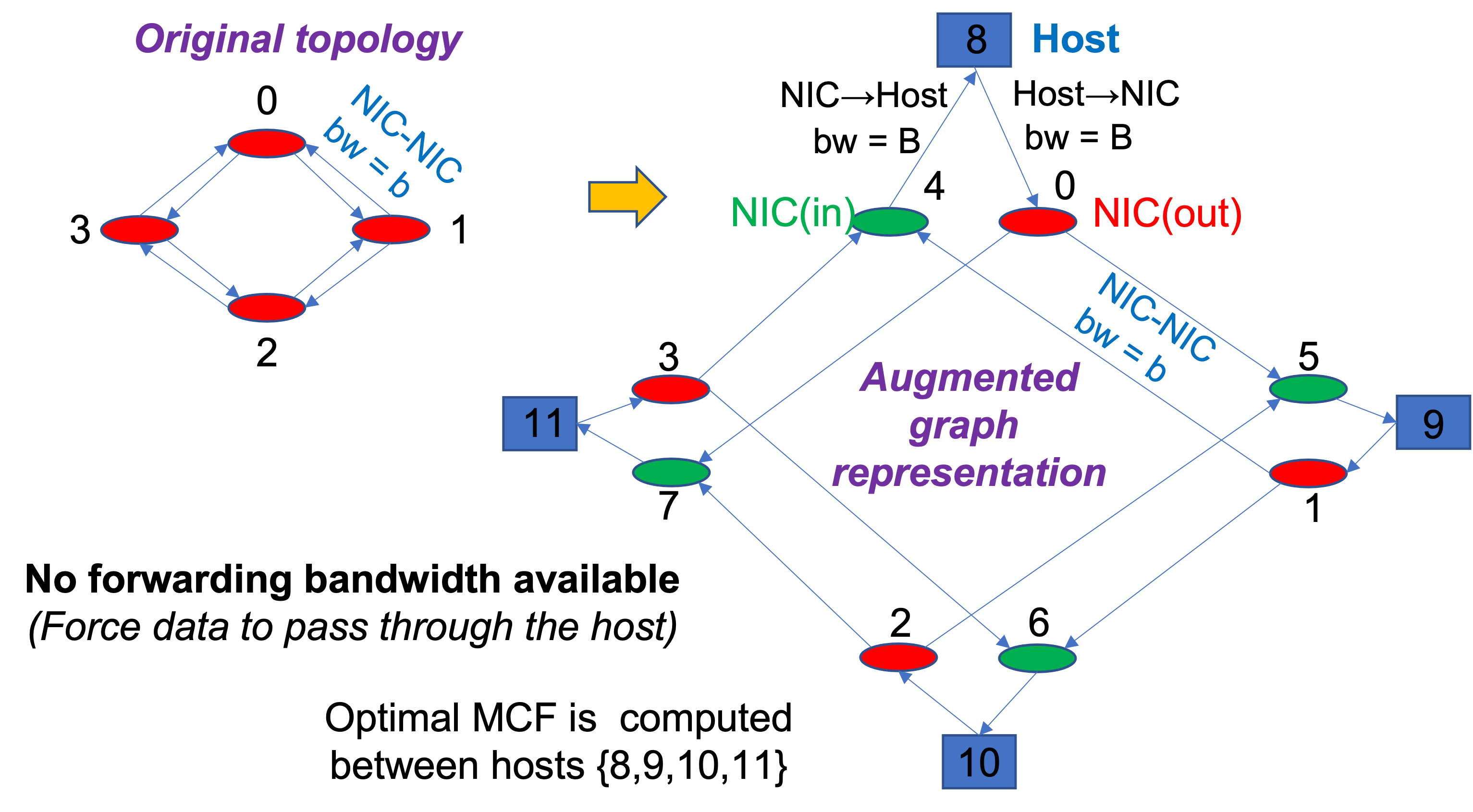}
\caption{Topology augmentation to model host-to-NIC bottleneck} %
\label{fig:with_fw_bw_bottlenecked}
\end{figure}
\vspace{-2ex}

%% file: impl.tex
\vspace{-2ex}
\section{Schedule Compilation}\label{sec:compiler}
We implement compilers and interpreters to lower the schedules and paths, %
and execute them on CPU and GPU runtimes.

\noindent {\bf Link-based Schedules:} The link-based MCF algorithm produces schedules that are chunked and lowered to both Microsoft's MSCCL~\cite{msccl} and Intel's oneCCL~\cite{oneccl}. 
MSCCL is an open-source collective communication library that extends NCCL~\cite{nccl} and RCCL~\cite{rccl} with an interpreter providing the ability to program custom collectives on GPUs.
Collective schedules are defined in XML as instructions (send/receive/reduce/copy) within GPU thread blocks that the interpreter executes.
We additionally lower the same schedules to oneCCL+libfabric~\cite{oneccl}, an open-source collective communications library by Intel that supports CPUs. 

We extended oneCCL with an interpreter, in a similar way that MSCCL extended NCCL, that executes the XMLs. 
The oneCCL XMLs similarly specify instructions (send, receive, reduce, copy, sync) and add scratch buffers for chunk forwarding.
The primary challenge in creating both MSCCL and oneCCL schedules was chunking.
The MCF solution produces the fractional rates $f_{(s,d),(u,v),t}$ for each commodity $(s,d)$ on each link $(u,v)$ at each time step $t$.
The lowering algorithm determines the smallest chunk size to support the lowest such rate, which guides how granularly a shard is chunked and how chunks are combined and forwarded by intermediate ranks.
At each time step, a rank then sends the total outgoing flow (the chunks received in the previous time step), receives the total incoming flow (chunks to receive in the current time step), and performs synchronization.
In both MSCCL and oneCCL, we have the ability to increase the number of channels by duplicating the schedule and running a parallel copy on different threads.
All sends and receives are asynchronous with no data dependencies.

\noindent  {\bf Path-based Schedules}. 
Recall that the weighted path-based MCF algorithm produces a set of weighted paths for each commodity (shard) $B_{i,j}$ (source $i$, dest $j$). 
Weighted multi-path routing and flow control should ideally be performed by the hardware, to which we would lower the MCF schedules (the per-commodity routes and weights).
In our testbeds, we use the Cerio (Rockport) NC1225 network card~\cite{rockportnic} (described in \S \ref{sec:testbeds}).
While the current version of the Cerio card natively supports multi-path routing on user-specified routes, it does not expose the capability to program the weights per route.
This means we cannot directly use the native multi-path capability, which we disable.
We instead approximate the weighted paths MCF by: (1) lowering the routes for each commodity $B_{i,j}$ to the underlying fabric, (2) dividing each shard/commodity into a set of equal-sized chunks, and (3) steering chunks onto routes as defined in our schedule. 

Specifically, the Cerio card exposes a utility for lowering our computed source routes to the hardware, where a route specifies the egress ports on the traversed links from source to destination as well as the layer identifier for the route; the layer identifier is used to assign routes to different virtual channels in order to eliminate deadlocks~\cite{rockport-deadlock}. 
The card also allows us to steer flows to routes at the application layer.
This is possible in ROCE v2 by setting the UDP source port when creating the RDMA Queue Pair (QP), such that the tuple (src port, src IP, dst IP, QP number) hashes to the desired route id. 
We implemented the scheduling and flow steering functionality in Open MPI+UCX~\cite{ompi}. 
The chunked schedule specification is lowered to an XML that is executed by our interpreter.
The latter is implemented in OMPI as part of the tuned collectives component within the Modular Component Architecture (MCA).
The schedule defines the chunks and path each should take, and the extended OMPI+UCX runtime creates the right number of QPs and performs the steering.
A shard is divided into a set of equal-sized chunks as follows: we compute the highest common factor across all path weights in the MCF solution and use that as the base chunk size (call it $c$). 
Each shard of size $m$ bytes is then divided into $\lceil m/c \rceil$ chunks, and the right number of chunks is assigned to each path based on the path weight.
This approach ensures all chunks (flows) fairly share the bandwidth, approximating the ideal MCF in practice on the Cerio fabric. 
We discuss the scalability limitations of this approach in \S\ref{sec:eval:discussion}.

%% file: eval.tex
\vspace{-3ex}
\section{Evaluation}\label{sec:eval}
We evaluate schedule performance and algorithm runtime at increasing scales on different hardware 
and on different direct-connect optical topologies, some of which are well-studied (complete bipartite, hypercube, twisted hypercube, Torus) while others are non-standard such as punctured tori with non-homogenous degree per node.
We also present performance results at a large scale in simulation. 
We summarize our performance results next.

 {\bf Performance of link-based schedules}:  Our lowered tsMCF schedules deliver near-optimal throughput performance on different topologies and scales, outperforming state-of-the-art baselines by up to 1.6$\times$ (\S\ref{sec:opt-coll} Fig.~\ref{fig:linkbasedxml}).

{\bf Performance of path-based schedules}: 
Our lowered MCF-extP and pMCF schedules deliver near-optimal throughput performance on different standard and non-standard topologies and scales, outperforming state-of-the-art scalable baselines by up to $30\%$ (\S\ref{sec:opt-coll}, Fig.~\ref{fig:routebased} and Fig.~\ref{fig:torus_punctured}), and the \ata speedups directly translate to speedups in the 3D FFT workload (\S\ref{sec:opt-coll}, Fig.~\ref{fig:fft_time}). MCF outperforms other scalable baselines at large scale in simulation (\S\ref{sec:large-sim}, Fig.~\ref{fig:compare_algos}, Fig~\ref{fig:topo_link_disable}).

{\bf Algorithm runtime}. Through large-scale simulations going up to 1000 nodes, our MCF decomposition approach yields orders of magnitude improvement in algorithm runtime over the original MCF and all other baseline schedule generation approaches (\S\ref{sec:large-sim}, Fig~\ref{fig:comptime_genkautz}).

{\bf Topology}. We identify GenKautz as a family of expander topologies that have near-optimal \ata performance while also having complete coverage in $N$ and $d$, outperforming other well-known expander topologies (\S\ref{sec:lower_bound}, Fig.~\ref{fig:genkautz_vs_lb_d4}).

\begin{figure*}
\centering
\includegraphics[width=0.98\textwidth]{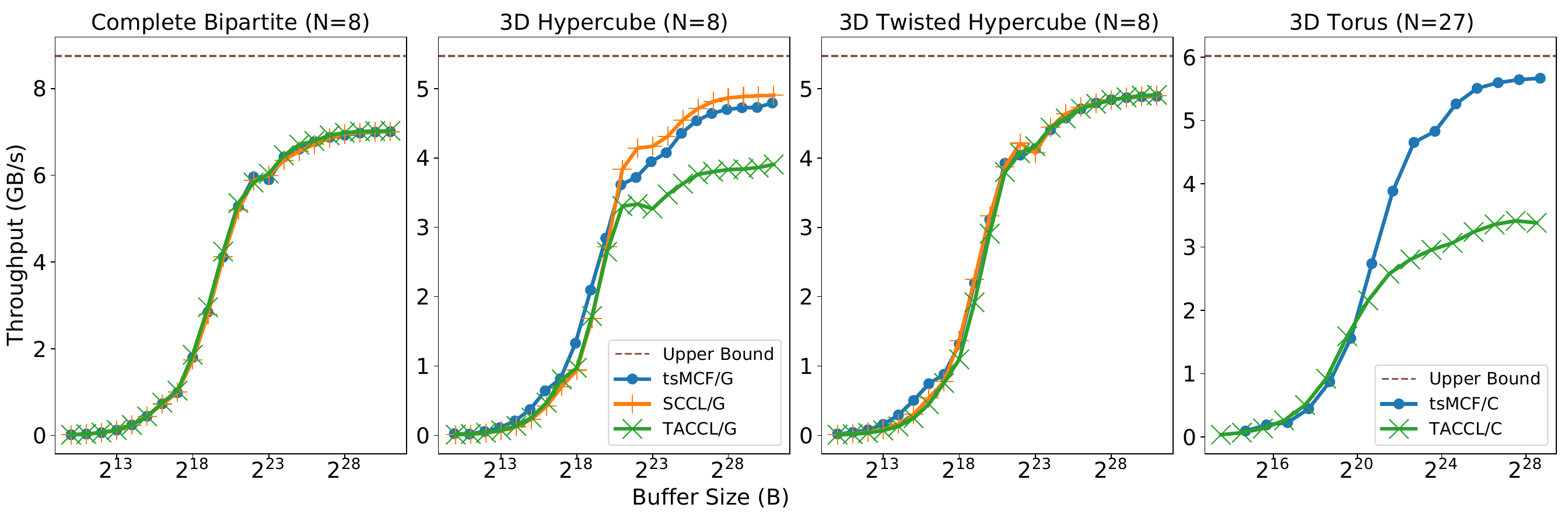}
\caption{Throughput of link-based \ata schedules on different topologies and runtimes. Appended \slash G indicates the schedule is lowered to GPUs and the MSCCL~\cite{msccl} runtime, whereas \slash C means the schedule is lowered to CPUs and the oneCCL~\cite{oneccl} runtime. Averaged over 20 iters.}
\label{fig:linkbasedxml}
\end{figure*}
\begin{figure*}
\centering
\includegraphics[width=0.98\textwidth]{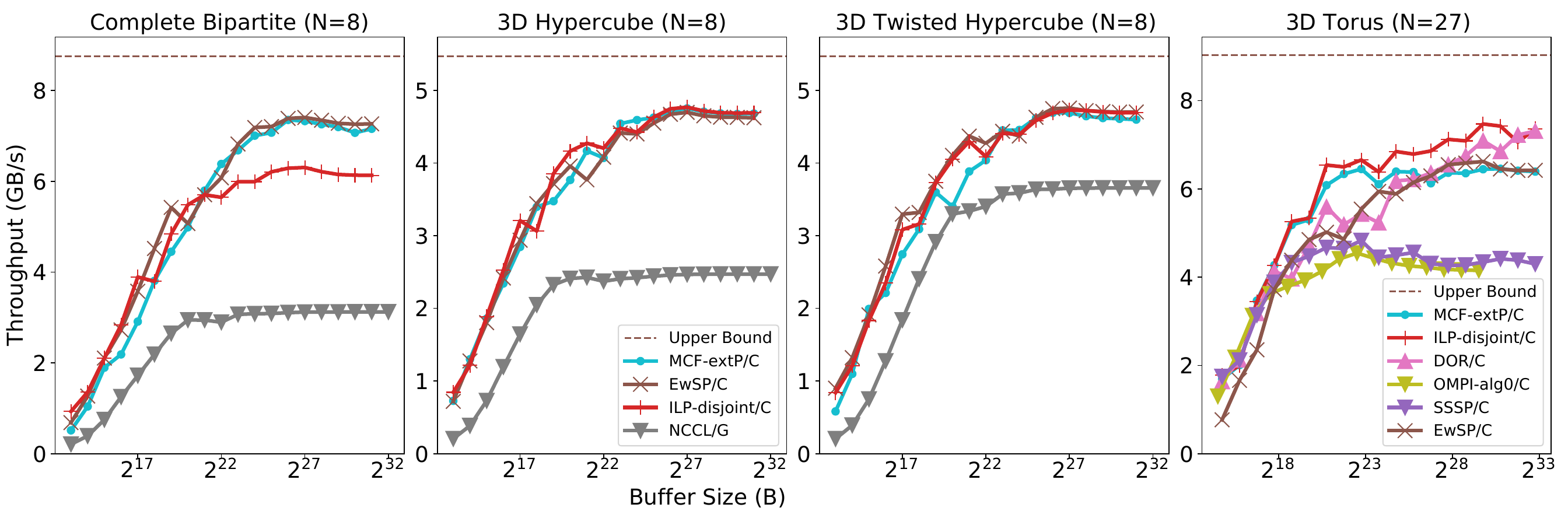}
\caption{Throughput of route-based \ata schedules on different topologies and runtimes. Appended \slash G indicates the schedule is lowered to GPUs and the NCCL~\cite{msccl} runtime, whereas \slash C means the schedule is lowered to CPUs and the Open MPI~\cite{ompi} runtime. Averaged over 20 iters.}
\label{fig:routebased}
\end{figure*}
\input{testbeds}
\subsection{Evaluation of optimized collectives}\label{sec:opt-coll}
\textbf{Time-stepped tsMCF schedules:}
Fig.~\ref{fig:linkbasedxml} shows our lowered link-based tsMCF schedules deliver near-optimal performance at large message sizes on different topologies and scales.
State-of-the-art scheduling algorithms such as SCCL~\cite{SCCL} do not scale, failing to terminate even at a modest 27-node scale (Fig.~\ref{fig:linkbasedxml}, right).
TACCL~\cite{TACCL} schedules, on the other hand, underperform on the Hypercube by $22\%$ and on the 3D Torus by up to 1.6$\times$ at large buffer sizes (Fig.~\ref{fig:linkbasedxml}, right).
Since our GPU testbed is constrained to an 8-node scale, we resort to the CPU cluster to evaluate at larger 27-node scale.
In Fig.~\ref{fig:linkbasedxml}, we append the algorithm name with \slash C to indicate that we lowered to oneCCL and ran on CPUs, whereas \slash G refers to lowering to MSCCL and running on the A100 GPUs.
Recall the link-based schedule experiments in Fig.~\ref{fig:linkbasedxml} do not use the routing capability from the underlying fabric; all transfers are on point-to-point links controlled by the scheduling thread(s) on the host GPU or CPU. 

On the 3D Torus (Fig.~\ref{fig:linkbasedxml}, right), tsMCF uses the bottleneck model (\S\ref{subsec:non_source-routed-fabrics} Fig.~\ref{fig:with_fw_bw_bottlenecked}) to produce the schedule since the host injection bandwidth (100 Gbps) is lower than the NIC bandwidth (150 Gbps for degree 6).  
The flow value produced by MCF on this bottlenecked 3D Torus topology is $f=\frac{2}{27}$, which dictates the theoretical upper bound of $(N-1)fb=(26)(\frac{2}{27})(3.125)=6.01$ GB/s.
On the other hand, the flow value in the non-bottlenecked setting, as we discuss shortly in Fig.~\ref{fig:routebased} (right), is $\frac{1}{9}$ which is $57\%$ higher.
The theoretical {\em upper bound} on throughput is $(N-1)fb$, where $f$ is the optimal flow value given by MCF (each commodity gets a max flow value of $f$ assuming link capacity is 1, 
so when link capacity is $b$, the max achievable flow out of a node sourcing $N-1$ flows is $(N-1)fb$).

In summary, our experiments show that {\em the optimized tsMCF schedules are ideal for both GPU and HPC fabrics where additional forwarding bandwidth is not available while being generalizable to a wide range of topologies}.

\begin{figure*}[t]
    \begin{minipage}{0.58\textwidth}
     \centering
        \includegraphics[width=0.95\columnwidth]{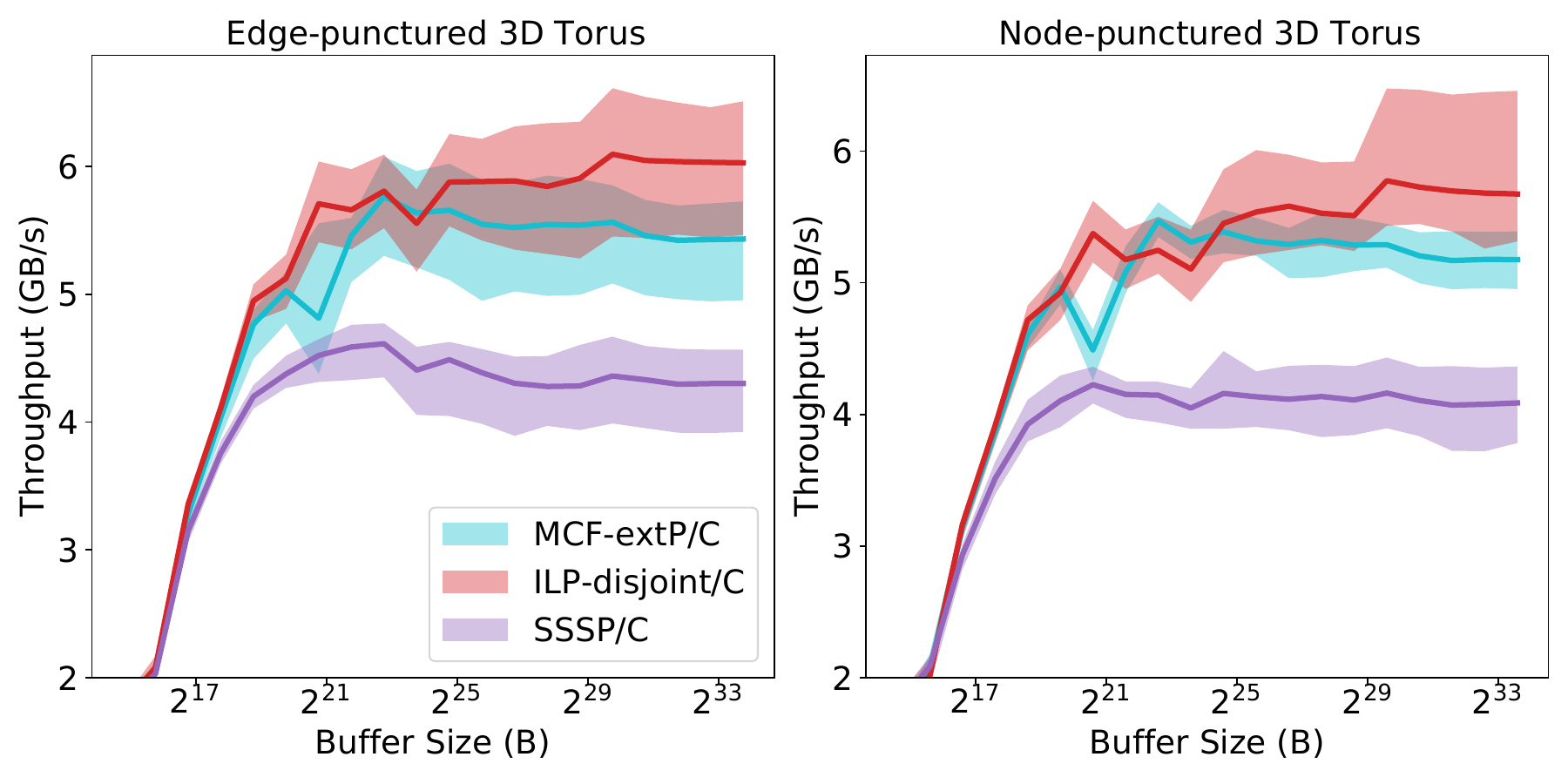}
	\caption{Performance on punctured 3D Torus, removing 3 links (left) or 3 edges (right) at random. Envelope min/max/average (line) over 10 instances (20 iters per)}
	\label{fig:torus_punctured}
    \end{minipage}
    \hfill
    \begin{minipage}{0.4\textwidth}
     \centering
        \includegraphics[width=0.98\columnwidth]{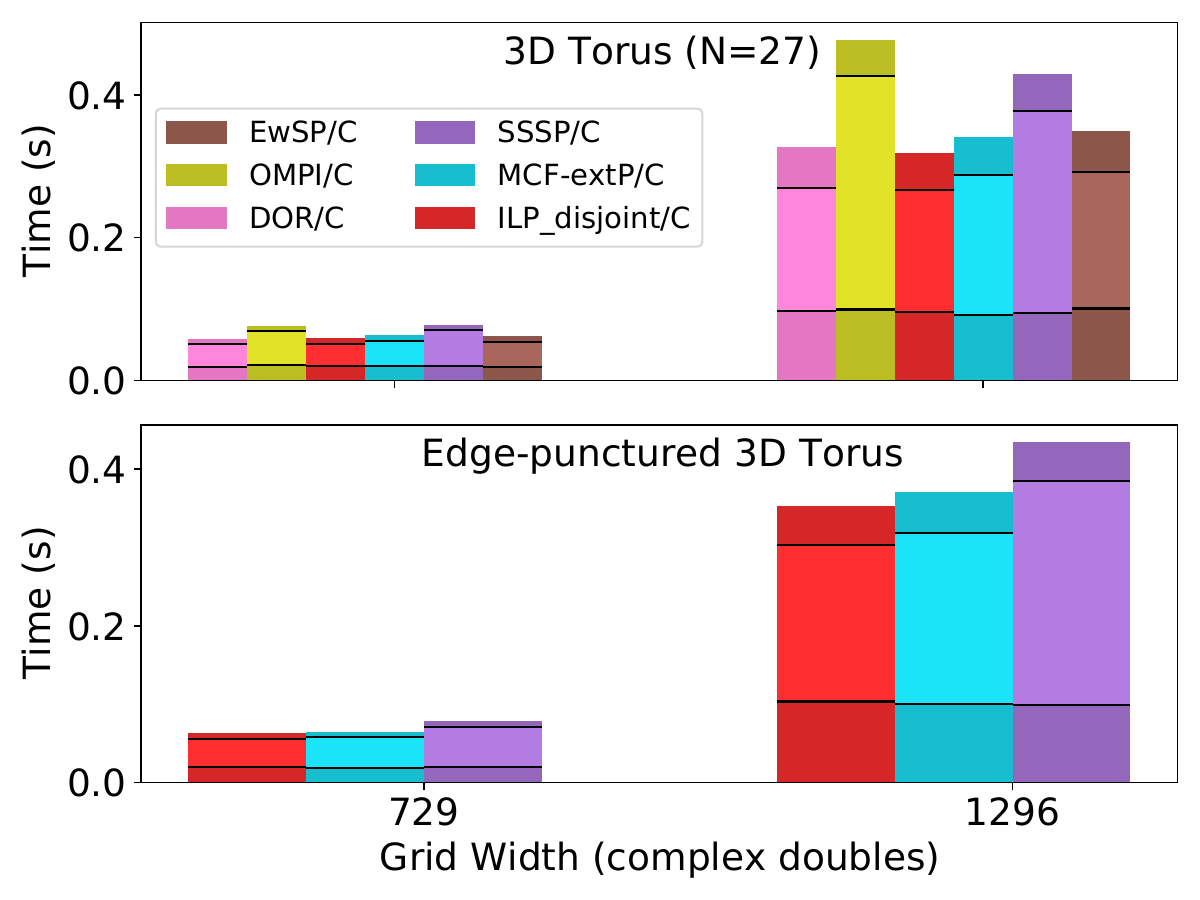}
	\caption{3D FFT Times ($N$=27 processes, 32 threads each)}
	\label{fig:fft_time}
    \end{minipage}
 \end{figure*}
\noindent
\textbf{Path-based schedules utilizing forwarding bandwidth}
We implement two link-load-minimizing single-path baselines. The first is based on Integer Linear Programming (ILP) which is tractable only at small scales.
It selects a subset of a candidate set of $(s,t)$ paths such that the maximum load on any link is minimized. Low maximum load leads to high \ata\! throughput. We experiment with both link-disjoint (ILP-disjoint) and shortest paths (ILP-shortest) in the candidate set.
The second baseline is Single Source Shortest Path (SSSP)~\cite{Domke11} heuristic that iteratively computes shortest paths through a graph whose link weights reflect the additional congestion caused due to each iteration.

Fig.~\ref{fig:routebased} shows our lowered MCF-extP schedules deliver near-optimal bandwidth performance in practice as we see at small scale (N=8) for the complete bipartite topology (left), and both the hypercube and twisted hypercube topologies (middle).  
On the complete bipartite, we see MCF-extP outperforms ILP-disjoint, which matches the theoretical results since the latter, which is constrained to a single path per commodity, is not bandwidth optimal on this topology. 
On the 27 node 3D Torus (right) on the supercomputer, our MCF-extP slightly underperforms due to practical limitations with injection rate control in the current fabric (more in \S\ref{sec:eval:discussion}).
Both dimension ordered routing (DOR)~\cite{dally1987deadlock} and ILP-disjoint are theoretically bandwidth optimal on the 3D Torus and are strong baselines. 
However, DOR does not work on non-Tori topologies, and ILP-disjoint does not produce optimal solutions in general and becomes intractable for larger topologies. 
SSSP is both scalable and general, but it produces sub-optimal solutions and is more than $50\%$ worse than MCF-extP on the 3D Torus at large buffers. 

MCF-extP also far outperforms NCCL and OMPI's native \ata algorithms up to $2.3\times$ on Bipartite, and $55\%$ on 3D Torus.
NCCL/OMPI native schedules perform N-1 point-to-point send/recv operations (flows) per rank.
These flows utilize the deadlock-free routes underneath which are computed by the Cerio fabric~\cite{rockport-deadlock}. 
We also implement and evaluate the Equal weight Shortest Path (EwSP) baseline that distributes each commodity equally on all the shortest paths between source and destination. 
While EwSP performs very well on all the four topologies in Fig.~\ref{fig:routebased}, this is not the case in general, as we show later in \S\ref{sec:large-sim}, Fig.~\ref{fig:compare_algos}.
 
Comparing path-based schedules of Fig.~\ref{fig:routebased} to link-based schedules of Fig.~\ref{fig:linkbasedxml} on the same topologies, we see that the bandwidth performance of MCF-extP is comparable to that of tsMCF at large buffer sizes for the Bipartite (degree=4), HyperCube and Twisted Hypercube (degree=3).
This is expected on these low-degree topologies since there is no additional forwarding bandwidth that path-based MCF can exploit, and both approaches have the same theoretical upper bound.
On the other hand, MCF-extP significantly outperforms tsMCF on the larger 27 node 3D Torus (by about $3.4\times$ at 1MB, and 15\% at larger buffers) since it is able to exploit the additional forwarding bandwidth in this case (degree=6, $B=$150 Gbps$\ge$ 100 Gbps injection bandwidth)--here MCF-extP is unable to reach its theoretical expected performance due to practical limitations on the current fabric (more in \S\ref{sec:eval:discussion}).
In general, we also see MCF-extP has a significant performance advantage at smaller buffers due to the superior latency performance of cut-through routing as compared to having to incur a global synchronization per timestep with tsMCF.

\noindent \textbf{Performance on non-standard topologies} We assess the topology-agnostic quality of MCF-extP on heterogeneous degree topologies produced by sampling sub-graphs from the 3D Torus. 
Specifically, we sample 10 different instances of the 3D Torus to create {\em edge-punctured} (3 random edges removed) and {\em node-punctured} (3 random nodes removed) Tori.
Baselines such as DOR are not defined for such punctured Tori. 
Fig.~\ref{fig:torus_punctured} shows the throughput of MCF-extP compared to ILP-disjoint and SSSP.
The results are consistent with those of Fig.~\ref{fig:routebased}, where MCF-extP significantly outperforms SSSP but underperforms ILP-disjoint due to practical limitations with injection rate control we discuss shortly in \S\ref{sec:eval:discussion}.
The results also match the empirical results, where we observe similar maximum link load for both ILP-disjoint and MCF (which is $\sim 30\%$ lower than SSSP).
Puncturing the Torus is a way to emulate failures of links and/or nodes, which would be expected in a cloud setting, and to show the superior performance of MCF in such scenarios.
When combined with the superior algorithm runtime (Fig.~\ref{fig:comptime_genkautz}), this shows that {\em our approach is able to more quickly react to failures in networks of hundreds of nodes without compromising on performance}.

\noindent \textbf{Workload speedups} We implement distributed 3D Fast Fourier Transform (FFT), and run it on the 27 node 3D Torus, and on the edge-punctured 3D Torus on up to $1296^3$ grid size, corresponding to \ata buffer size up to 1.29 GB.
Fig.~\ref{fig:fft_time} shows the speedups in the 3D FFT on the 3D Torus (top, corresponding to \ata speedups from Fig.~\ref{fig:routebased}, right) and the punctured 3D Torus (bottom, corresponding to \ata speedups from Fig.~\ref{fig:torus_punctured}).
We observe up to $20\%$ ($14.9\%$) total speedup in the FFT time using our MCF schedules compared  to SSSP schedules on the 3D Torus (edge-punctured 3D Torus).
The speedups are a direct result of the faster \ata MCF schedules, consistent with Fig.~\ref{fig:routebased} and Fig.~\ref{fig:torus_punctured}.
We use the latest version of the FFTW library~\cite{fftw} with multi-threading and with OMPI running our \ata schedules. 
We use slab decomposition, where each process (1 multi-threaded process per node) performs three steps: first computes 2D FFTs on its slabs and packs the data, then runs \ata with all other processes, and finally unpacks and computes 1D FFTs to complete the 3D FFT.
These three steps involved in the FFT computation are shown with bands in the bars of Fig.~\ref{fig:fft_time} with the first step corresponding to the bottom band.
\input{eval_scale}

%% file: testbeds.tex
\vspace{-2ex}
\subsection{Direct-connect testbed and cluster}\label{sec:testbeds}

We evaluate the \ata schedules on two testbeds: an internal 8 server (1 NVIDIA A100 GPU~\cite{gpu} per server) testbed that supports topology reconfiguration, and an external 27 server (1 CPU per server) cluster at the Texas Advanced Computing Center (TACC)~\cite{tacc, tacc-rockport} where topology is fixed to the Torus.

\noindent {\bf Cerio Card}. In both testbeds, each server is equipped with a Cerio NC1225 network card~\cite{rockportnic}.
The card supports source routing with multi-path and cut-through flow control, 
and stores up to 8 routes per destination.
It offers up to 300 Gbps of total forwarding bandwidth using 12x 25 Gbps links ($b=25$ Gbps or $3.125$ GB/s), and supports 100 Gbps of injection bandwidth from the host or GPU using x16 PCIe gen3. 
We can accordingly evaluate both link and path-based schedules on both testbeds.

\noindent {\bf Internal GPU Testbed}. 
The network cards are directly connected via a Telescent optical patch panel~\cite{patchpanel}.
Our testbed can realize different topologies by reconfiguring the patch panel. 
We limit our evaluation to bidirectional topologies, specifically hypercube and twisted hypercube both with degree 3 (i.e., $B$=75 Gbps), and complete bipartite with degree 4 ($B$=100 Gbps).
While unidirectional topologies (e.g., GenKautz) can be realized by configuring the patch panel in simplex mode, we cannot accurately evaluate their performance since the requisite overlay 
routing for the reverse path traffic (acks) is currently only supported using routing rules performed by the kernel leading to unpredictable RTTs. 

\noindent {\bf TACC HPC Cluster}. A cluster of CPU servers at TACC are connected in a fixed torus topology using the Cerio fabric~\cite{tacc, tacc-rockport}.
We use a 3x3x3 torus (27 nodes, degree=6) from the cluster to run our experiments.
The host injection bandwidth of 100 Gbps is less than $B$=150 Gbps for degree 6; hence, we evaluate the benefits of path based schedules to exploit the extra forwarding bandwidth for \ata\!.

%% file: eval_scale.tex
\subsection{Large scale numerical simulations}\label{sec:large-sim}
\begin{figure}[b]
     \centering
        \includegraphics[width=0.9\columnwidth]{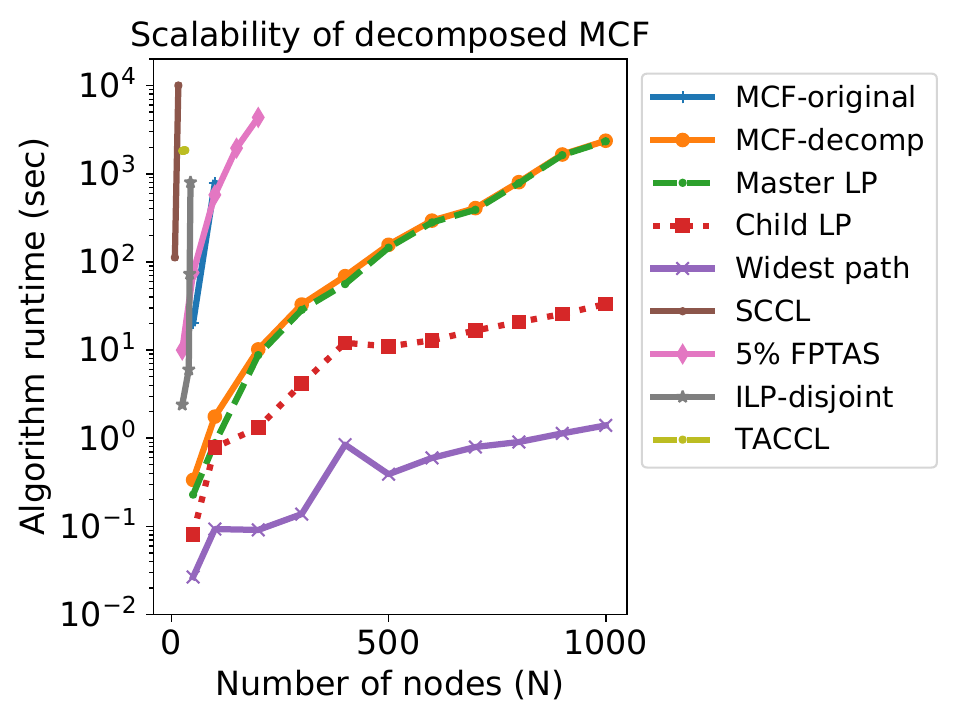}
	\caption{Scaling on GenKautz (degree=4); MCF-decomp is master LP + $N$ parallel child LPs + widest path extraction heuristic}
	\label{fig:comptime_genkautz}
\end{figure}
\textbf{Benefits of MCF decomposition:} 
Fig.~\ref{fig:comptime_genkautz} shows that our MCF decomposition approach (MCF-decomp) yields orders of magnitude improvement in algorithm runtime over the original MCF (MCF-original), and the other highly unscalable topology-agnostic schedule generation approaches such as SCCL~\cite{SCCL}, TACCL~\cite{TACCL}, Karakostas' Fully Polynomial Time Approximation Scheme (FPTAS)~\cite{karakostas2008faster}, and ILP-disjoint.
We show the computation time required for MCF-original (link-based) on $N^3$ variables and MCF-decomp on $N^2$ variables for the Generalized Kautz graph~\cite{genkautz} (also see \S\ref{sec:lower_bound}) 
for various choices of $N$. We observe that even at the scales of $N=50,100$, MCF-decomp is two orders of magnitude faster, while MCF-original fails to produce a solution for $N>100$ nodes in a reasonable time.
In contrast, SCCL~\cite{SCCL} is unable to generate \ata schedules for $N=16$ even in $10^4$ seconds. 
TACCL~\cite{TACCL}, a heuristic designed to be more scalable than SCCL, takes over 30 minutes to generate \ata schedules for even 32-node networks. 
Furthermore, for slightly larger networks, it is unable to produce even an approximate solution in 30 minutes. 
Even ILP-disjoint is only able to generate optimal schedules up to $N=44$.
While we evaluate on the GenKautz topology in Fig.~\ref{fig:comptime_genkautz}, the scaling trends apply to other topologies (expanders, tori) as well.

The algorithm runtime of MCF-decomp follows a polynomial trend, and it is dominated by the time taken to solve the ``Master LP", which is a source-based MCF leading to a solution in 40 minutes for $N=1000$ provided all ``Child LPs'' and subsequent ``Widest path" extraction functions are run in parallel on $N$ cores. The exact degree of the polynomial governing the speedup factor over original MCF is determined by that of the polynomial that determines the time complexity of solving such network flow LP problems. For example, if a state-of-the-art LP solver takes $\mathcal{O}(N^{2.37})$ time to solve $N$-variable problems, the speedup factor can be estimated to be $\mathcal{O}(N^{3\times2.37}/N^{2\times2.37}) = \mathcal{O}(N^{2.37})$.
 
Fig.~\ref{fig:comptime_genkautz} also shows the runtime performance of a state-of-the-art FPTAS by Karakostas~\cite{karakostas2008faster} at $\epsilon = 0.05$. While it shows a polynomial scaling trend unlike the other baseline schemes mentioned earlier, it significantly underperforms the MCF (decomposed) algorithm in running time even after sacrificing optimality, thus making it impractical for networks larger than a couple of hundred nodes. The FPTAS, being an inherently sequential algorithm, cannot exploit the high degree of parallelism that is exploited by decomposed MCF.
\begin{figure}[htbp]
\centering
	\includegraphics[width=0.9\columnwidth]{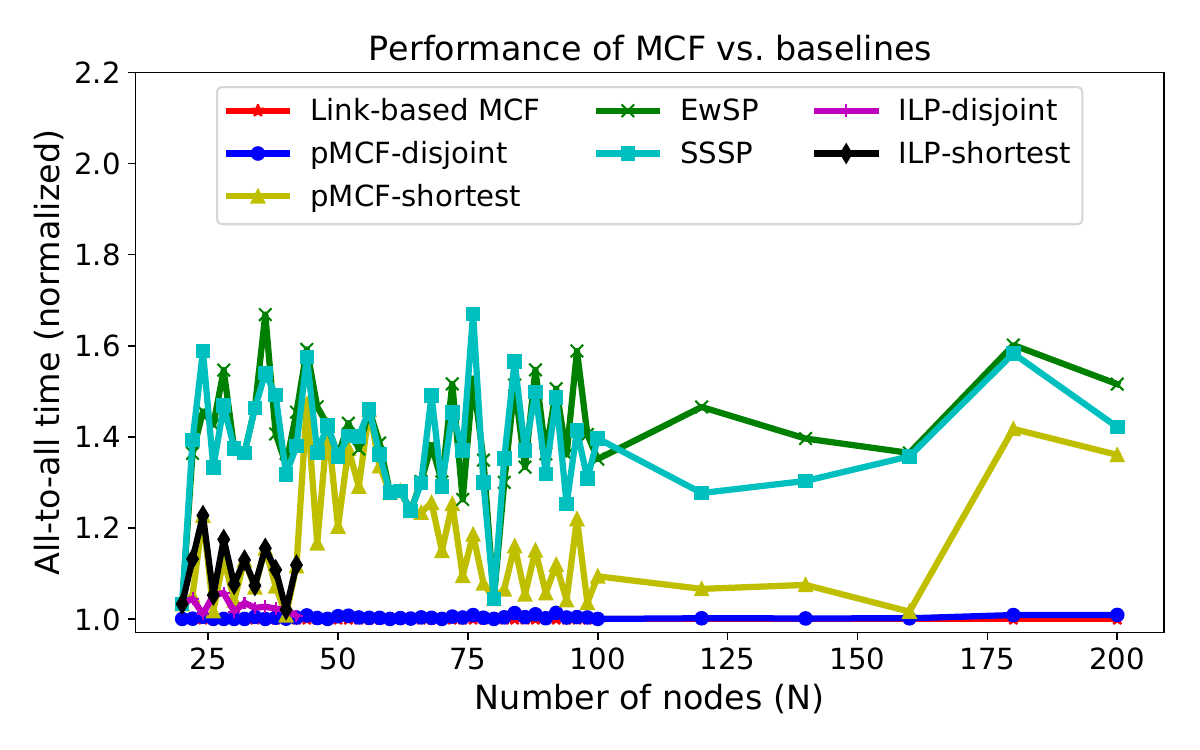}	
	\caption{Performance on degree 4 Generalized Kautz graphs normalized by Link-based MCF}
	\label{fig:compare_algos}
\end{figure}

\noindent \textbf{Performance of schedules that utilize NIC forwarding bandwidth:}
Fig. \ref{fig:compare_algos} (and Fig.~\ref{fig:topo_link_disable}) shows the {\em \ata time} (the time to concurrently transmit the workload assuming unit commodity demand and link capacities; it is equal to $1/MCF$ and the maximum link load) of various path-based schemes normalized by optimal link-based MCF.
Comparing single path schemes, mainly ILP and SSSP, although SSSP is fast, it is up to $1.6\times$ worse than the theoretically optimal MCF solution.
On the other hand, the ILP schemes, although performant, do not scale well (Fig.~\ref{fig:comptime_genkautz}).
\begin{figure*}[t!]
     \begin{minipage}{0.32\textwidth}
        \centering
	\includegraphics[width=1\columnwidth]{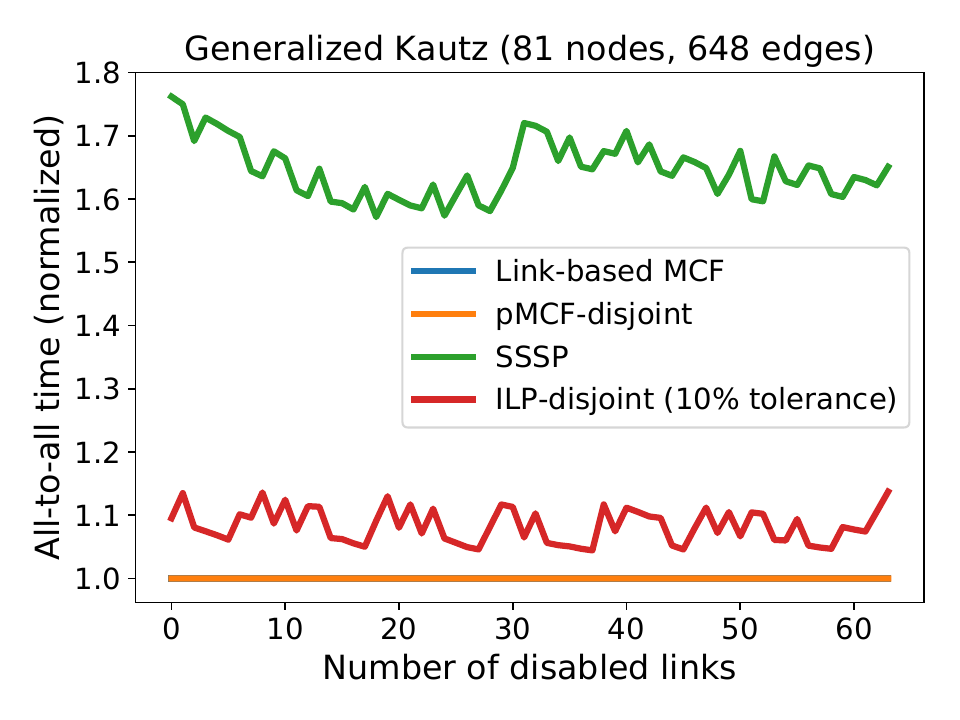}
	\caption{Performance on $N$=81 Gen Kautz w/ links disabled normalized by Link-based MCF} %
	\label{fig:topo_link_disable}
    \end{minipage}
    \hfill
     \begin{minipage}{0.65\textwidth}
        \centering
	\includegraphics[width=0.495\columnwidth]{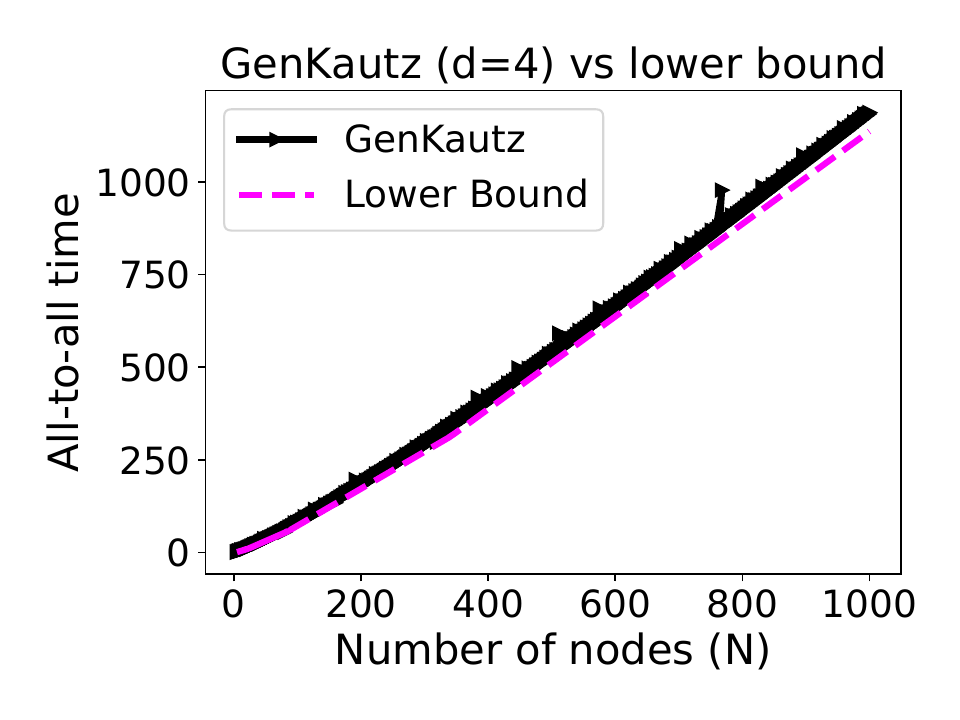}
	\includegraphics[width=0.495\columnwidth]{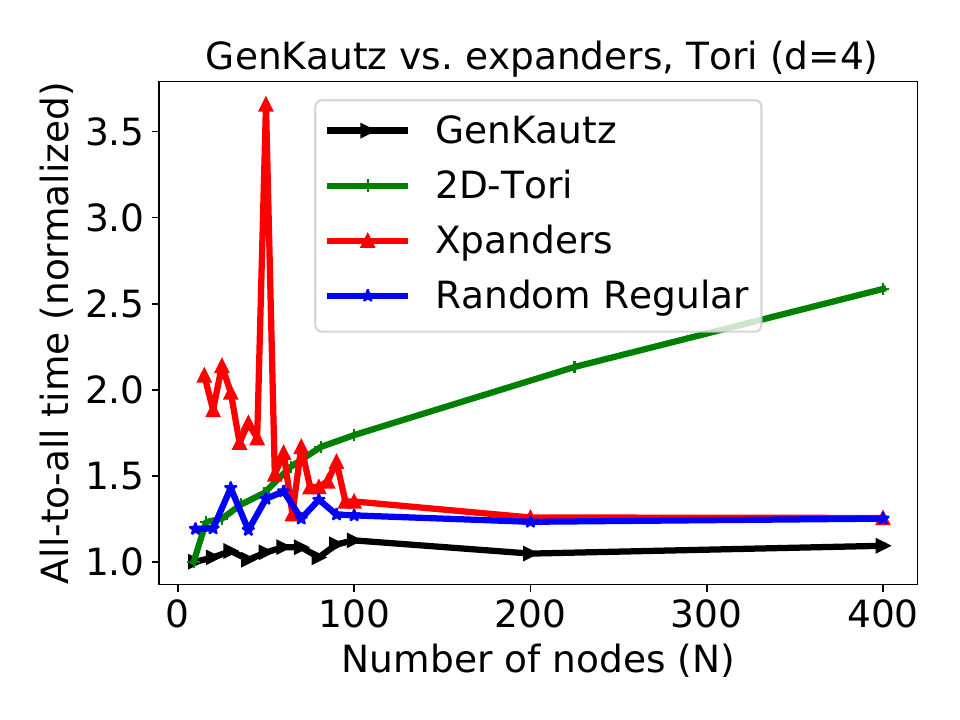}
	\caption{GenKautz topology performance relative to lower bound (left), and expanders and 2D-Tori (right, normalized by lower bound)}
	\label{fig:genkautz_vs_lb_d4}
    \end{minipage}
\end{figure*}

Although multipath schemes exploit the path diversity of the network well, naive multipath approaches such as Equal Weight Shortest Path that distribute the workload evenly along all available $(s,t)$ shortest paths do not perform well. Its performance is similar to that of SSSP. However, \emph{weighted} multipath schemes such as pMCF (Path-based MCF (disjoint)) and MCF-extP (Link-based MCF with path extraction) can achieve optimal or near-optimal performance. 
pMCF is optimal in theory if the number of $(s,t)$ paths in the initial set is not restricted. However, this set can be extremely large, thus necessitating the development of decomposed MCF-extP, whose scalability has been illustrated in Fig. \ref{fig:comptime_genkautz}. However, in practice, if the $(s,t)$ path set is restricted to link-disjoint paths, pMCF can almost match the optimal performance of pure link-based MCF. Interestingly, pMCF run with all shortest $(s,t)$ paths exhibits suboptimal  performance, especially for expander graphs (Fig. \ref{fig:compare_algos}) since the latter do not have too many shortest paths. While pMCF (with shortest paths) performs well on topologies such as tori, the starting path set is often exponentially large in size, thus making the scheme impractical. In contrast, link-disjoint $(s,t)$ paths can be computed in polynomial time for arbitrary topologies. Since there are at most $d$ disjoint paths for any $(s,t)$ pair, the number of LP variables is $\mathcal{O}(N^2)$, making it comparable to decomposed Link-based MCF in terms of time complexity.

Fig. \ref{fig:topo_link_disable} illustrates the generality and good performance of MCF schemes (when compared to baselines like SSSP) since they perform optimally/near-optimally on heterogeneous degree-irregular subgraphs, e.g., formed by disabling random links. At the $N=81$ scale considered ($d=8$), interestingly, ILP-disjoint shows performance close to link-based MCF's if we allow it a tolerance factor ($\epsilon=0.1$); but it does not scale to large $N$ owing to it being NP-hard. 
\subsection{Near throughput-optimal \ata\! topologies}
\label{sec:lower_bound}
We ask \textit{which topologies with $N$ nodes and degree $d$ yield the best \ata\! performance?}
This is an important design question for supercomputing clusters and is becoming more important with increased deployments of reconfigurable fabrics~\cite{wang2023topoopt, jouppi2023tpu, lightwave-sigcomm}.
We derive an upper bound for \ata\! throughput (equivalently, lower bound the \ata\! collective time) for arbitrary $d$-regular graphs with $N$ nodes, and then highlight an expander graph that achieves performance close to this bound.
Since the bound is reasonably tight, it allows us to compare the performance of various topologies with respect to the theoretical optimal.

\begin{theorem}[Lower bound on \ata\! time.]\label{thm:lb}
The time taken to accomplish all-to-all communication in a $d$-regular graph $G$ on $N$ nodes scales as $\Omega(N\log_d N)$.
\end{theorem}
\begin{proof}[Proof sketch]
First, consider a single source node $r$ and an arborescence (outgoing rooted directed tree) of $G$ (denoted by $T_{d,N}$), which has $N$ nodes and maximum out-degree $d$. It has $d^k$ number of nodes at all levels $k$ except when $k$ is equal to its height. If $r$ needs to send $N-1$ flows of value $f$ to each of the other $N-1$ nodes along $T_{d,N}$, the minimum capacity needed is $f\times\sum_{u\in V_{T_{d,N}}}D(r,u)$, where $D(s,t)$ is the distance (hop count) between $s$ and $t$ along $T_{d,N}$. Thus, the minimum capacity needed for sending commodities from all nodes (along $N$ respective rooted arborescences) is $N\times f\times\sum_{u\in V_{T_{d,N}}}D(r,u)$. Assuming that a $d$-regular di-graph has $d$ outgoing unit capacity links per node, the total capacity available in the network is $d\times N$. It follows that $f\times\sum_{u\in V_{T_{d,N}}}D(r,u) \leq d$. The \ata\! workload completion time/latency is given by:
\begin{align}\label{eq:lb}
1/f \geq \sum_{u\in V_{T_{d,N}}} D(r,u) / d. 
\end{align}
Also, in a $d$-regular graph where each arborescence is a full $k$ layer tree, $N=\sum_{i=0}^{k-1}d^i=\frac{d^k-1}{d-1}$ and
$\sum_{u\in V_{T_{d,N}}}D(r,u)=\sum_{i=0}^{k-1}i\times d^i =\frac{d^{k+1}(k-1)-d^kk+d}{(d-1)^2} = \Theta(k\:d^{k-1})$. The RHS of Eq. \eqref{eq:lb} then becomes $\Theta(k\:d^{k-2}) = \Theta(N\log_d N)$. This establishes the scaling law for the lower bound on \ata\! time in any $d$-regular $N$-node topology.
\end{proof}

\noindent\textbf{Graphs achieving \ata\! time bound} Generalized Kautz graphs~\cite{genkautz} constitute a family of \emph{expander graphs} that comes close to achieving the bound in Theorem \ref{thm:lb}. A benefit of these graphs is that we can generate an instance for \emph{any} value of $N$ and $d$. Such coverage in $N$ and $d$ is not possible for most graph families popular in the HPC community, such as mesh, tori, SlimFly~\cite{besta2014slim}, SpectralFly~\cite{young2022spectralfly}, etc. 
Fig.~\ref{fig:genkautz_vs_lb_d4}(left) shows the simulated \ata\! performance of the GenKautz graph for $d=4$ with respect to the derived lower bound.
We observe that these graphs get very close to the \ata\! lower bound for any $d$-regular graph (with the ratio between the two approaching $1$ for large $N$) thus making them optimal expanders for \ata\! collectives.

Fig.~\ref{fig:genkautz_vs_lb_d4} (right) compares the performance of GenKautz with non-expanders (e.g., 2D-tori) and other well-known expander graphs, e.g., Xpander~\cite{valadarsky2016xpander} and random regular graphs /Jellyfish~\cite{singla2012jellyfish}. The last two are also known to exhibit good coverage in $N$ and $d$. While the expanders significantly outperform non-expanders (for $d=4$, GenKautz has about $2.4\times$ lower latency than 2D-tori for large $N$), GenKautz has the best performance among the expanders (e.g., it is 10\% better than Xpander and random regular graphs). Similar trends are observed for higher degrees.

\vspace{-2ex}
\subsection{Conclusion and Discussion}\label{sec:eval:discussion}
We develop efficient schedules and topologies for \ata collective communications geared for large-scale direct-connect fabrics.
Our results demonstrate that the time-stepped MCF approach and compiler are highly scalable and achieve near-optimal bandwidth performance in practice. 
This is valuable for many applications that rely on \ata when run on GPU (or CPU) clusters, such as deep learning training and inference.
For path-based MCF-extP and pMCF, our results similarly demonstrate excellent performance in practice at the smaller 8 and 27-node scales.
Path-based MCF is able to exploit additional forwarding bandwidth to speed up the collective. 
Our experiments with MCF, however, uncovered additional theoretical and practical challenges, which we are addressing as part of our ongoing (and future) work.\\
{\bf Deadlock-free routing}:  When lowering path-based MCF routes to fabrics with wormhole (flit-based) routing such as Cerio, routes must be deadlock-free~\cite{dally1987deadlock}.
We implemented several variants of common algorithms for breaking deadlocks, such as DF-SSSP~\cite{Domke11} and LASH~\cite{skeie2002layered}, which assign virtual channels to routes after the routes are computed. 
We found that a variant of LASH~\cite{skeie2002layered}, which we call LASH-sequential performed best in terms of requiring the least number of layers; specifically, it required no more than 4 layers across all the algorithms (MCF, ILP, EwSP, etc.) and topologies we evaluated. %
Minimizing the number of VCs to make a given set of routes deadlock free is NP-hard~\cite{Domke11}. 
An open question is {\em how to generate \ata schedules that optimize throughput while ensuring deadlock freedom}. 

\noindent{\bf Injection rate control}: On the practical front, a limitation of path-based MCF solutions when lowered to existing fabrics is support for {\em injection rate control}.
As described in \S\ref{sec:compiler}, we implemented an approximation of injection rate control by splitting shards into granular equal-sized chunks/flows and steering them on routes.
While this approach worked very well at an 8-node scale as a proof of concept, and achieved reasonable performance at a larger 27-node scale, it is in general not scalable.
Granular chunking significantly increases the total number of active Queue Pairs (QPs) in the network. %
And our \ata micro-benchmarking experiments clearly showed a reduction in the achievable per-flow bandwidth as the number of QPs increased on the Cerio fabric, likely due to increased contention. %
We are pursuing two approaches to address this challenge: (1) introduce time steps into the routed MCF schedules and partition the flows across multiple timesteps, %
and (2) work with the Cerio vendor on options to expose injection rate control in the hardware. 

\noindent{\bf Clustered/Hybrid Configurations}: We are extending the general MCF formulation and the implementation to handle hybrid clustered settings with possibly severe imbalance between internal link bandwidth within a server, and external bandwidth (e.g., several Tbps internal bandwidth vs several Gbps external bandwidth in GPU servers from NVIDIA and AMD), and with possibly internal switching.

\noindent{\bf Future work}: Our solution computes a static schedule that assumes dedicated underlying link bandwidth for the duration of the \ata collective.
This is a reasonable assumption in several direct-connect settings where a cluster manager allocates circuits to jobs on a non-interfering basis.
Handling more dynamic environments with multiple jobs contending for bandwidth is left for future work.

%% file: ack.tex
\vspace{-3ex}
\section{Acknowledgement}
We thank the Rockport Networks team, Matthew Williams, Doug Taylor, Nick Tkotz, and Shaun Hennesey for their help and insights with the Rockport fabric experiments, and with making their slice of the TACC supercomputer available to us.
We also thank Amit Ruhela for his help and insights.\\
\noindent This research was developed with funding from the Defense Advanced Research Projects Agency (DARPA) under contract number HR001120C0089. 
The views, opinions and/or findings expressed are those of the authors and should not be interpreted as representing the official views or policies of the Department of Defense or the U.S. Government.